\documentclass{ifacconf}
\usepackage{natbib}
\usepackage{graphics} % for pdf, bitmapped graphics files
\usepackage{epsfig} % for postscript graphics files
\usepackage{times} % assumes new font selection scheme installed
\usepackage{amsmath} % assumes amsmath package installed
\usepackage{amssymb}  % assumes amsmath package installed
\usepackage{amsfonts}
\usepackage{color}
\usepackage{algorithm,algpseudocode}
\usepackage{soul}
\usepackage{tikz}
\usepackage{tikz-network}
\usepackage{csquotes}
\usetikzlibrary{backgrounds}
\tikzset{background rectangle/.style={fill=yellow!30}}
\usetikzlibrary{arrows.meta, positioning}
\usepackage{url}

\usepackage{enumitem}  
\setlist[description]{leftmargin=\parindent}
\usepackage[mathscr]{eucal}
\usepackage{subfigure}
\usepackage{comment}

\def\diag{{\rm diag}}

\def\im{{\rm Im}}
\newcommand{\Cm}{\mathcal{C}}
\newcommand{\Dm}{\mathcal{D}}
\newcommand{\Em}{\mathcal{E}}
\newcommand{\Gm}{\mathcal{G}}

\newcommand{\Imin}{\mathcal{B}}

\newcommand{\Pm}{\mathcal{P}}

\newcommand{\Sm}{\mathcal{S}}
\newcommand{\Tm}{\mathcal{T}}
\newcommand{\Vm}{\mathcal{V}}
\newcommand{\Wm}{\mathcal{W}}

\newcommand{\Zm}{\mathcal{Z}}

\newcommand{\beq}{\begin{equation}}
\newcommand{\eeq}{\end{equation}}
\newcommand{\beqa}{\begin{eqnarray}}
\newcommand{\eeqa}{\end{eqnarray}}
\newcommand{\beqan}{\begin{eqnarray*}}
\newcommand{\eeqan}{\end{eqnarray*}}

\newcommand{\bite}{\begin{itemize}}
\newcommand{\eite}{\end{itemize}}
\newcommand{\benu}{\begin{enumerate}}
\newcommand{\eenu}{\end{enumerate}}

\newcommand\rednote[1]{\textcolor{red}{#1}}

\definecolor{darkpastelgreen}{rgb}{0.01, 0.75, 0.24}

\usepackage{placeins}
\usepackage{bm}
\usepackage{lipsum}
\algnewcommand{\LeftComment}[1]{\Statex \(\triangleright\) #1}

\begin{document}
\begin{frontmatter}

\title{Minimal Input Cardinality Disturbance Decoupling of Coupled Oscillators via Output Feedback with Application to Power Networks\thanksref{footnoteinfo}} 

\thanks[footnoteinfo]{Work supported in part by grants from the Swedish Research Council (grant n. 2020-03701 to C.A.) and by the ELLIIT program at Link\"oping University.}

\author[First]{Luca Claude Gino Lebon} 
\author[Second]{Johan Lindberg} 
\author[First]{Claudio Altafini}

\address[First]{Department of Electrical Engineering, Division of Automatic Control,
        Linkoping University, SE-581 83 Linköping, Sweden
        (e-mails: luca.lebon@liu.se, claudio.altafini@liu.se)}
\address[Second]{Department of Electrical Engineering, Division of Automatic Control, Lund University, SE-221 00 Lund, Sweden
        (e-mail: johan.lindberg@control.lth.se)}

\begin{abstract}                % Abstract of 50--100 words
In this paper, we identify the smallest set of control input nodes and an associated output feedback law that achieves complete disturbance decoupling for a class of coupled oscillator networks. The focus is specifically on systems linearized around a stable phase-locked synchronized state. The proposed theoretical framework is applied to the linearized swing dynamics of power grids operating near synchronization. In this context, the disturbance decoupling problem corresponds to isolating subsets of nodes from exogenous disturbances by means of batteries that can both add or withdraw active power. Numerical simulations carried out on the IEEE New England 39-bus system show that the proposed methodology not only yields a minimal actuator placement ensuring effective disturbance rejection, but also preserves the internal stability of the closed-loop system.
\end{abstract}

\end{frontmatter}

%%%%%%%%%%%%%%%%%%%%%%%%%%%%%%%%%%%%%%%%%%%%%%%%%%%%%%%%%%%%%%%%%%%%%%%%%%%%%%%%

\section{Introduction}
The control of power networks, or more generally of coupled oscillators, has attracted significant attention in the systems and control community due to its broad practical relevance (\cite{Glover,skardal2015control,Mach20}). Of particular interest is the system behavior under synchronization, i.e., when all agents rotate at a common frequency. Considerable research has been devoted to achieving stable synchronization in inductive microgrids using DC/AC droop-controlled inverters (\cite{simpson2013synchronization,dorfler2013synchronization}). Moreover, game-theoretic analyses have been employed to study the loss of coherence in synchronized oscillator networks (\cite{yin2011synchronization}).

Building on these results, this paper assumes that a stable synchronized state has been reached, allowing the system to be linearized around such an equilibrium. In this regime, the dynamics can be represented in state-space form, or more generally, in descriptor form. Within this framework, prior work has focused on attack detection and identification, where attacks are modeled as disturbance signals acting on specific nodes of the network (\cite{pasqualetti2013attack}). In this study, we assume that the attack signature, i.e., the location of the disturbances in the network, is known.  We design a controller that decouples a subset of target nodes (those to be protected) from the influence of the disturbances. Related research on disturbance rejection for lossy and lossless power systems has been conducted using $H_2$ and $H_\infty$ norm techniques (\cite{lindberg2022fundamental,Pates22}). Other works include active power injection to damp inter-area oscillations (\cite{Bryne14}). Here, we extend this line of work by selecting a minimal number of control inputs required to achieve decoupling, via output feedback from specific measurement nodes.
%, under the assumption that the control law is an output feedback from selected measurement nodes. 
In large-scale power systems, ensuring that local disturbances do not propagate to critical or protected regions is essential for maintaining stability and reliability. Traditional protection schemes, such as disconnecting transmission lines or isolating generators, can effectively contain disturbances but often at the expense of service continuity and operational efficiency. The proposed approach enables the real-time decoupling of disturbance propagation between network regions without interrupting nominal operation. This is achieved through the strategic selection of a minimal set of control nodes, which act as intervention points to compensate for the spread of the disturbances. By ensuring that only the smallest possible number of nodes are actuated, the method achieves disturbance rejection in the least invasive and most resource-efficient manner. The proposed control framework builds upon classical results in geometric control theory, reinterpreted and analyzed in the context of networked systems as discussed in~\cite{lebon2026geometric_tac}.
%This paper relies heavily on the results of~\cite{lebon2026geometric}, which are here quoted without proofs (and sometimes without statements) for lack of space. 

The remainder of the paper is organized as follows. Section~II provides background on the linearization of coupled oscillator networks around a synchronized equilibrium and introduces the Minimal Input Cardinality Disturbance Decoupling Problem. Section~III presents the main results, including the procedure for computing the optimal input set and the construction of the output feedback matrix. Section~IV illustrates the application of the proposed methods through a comprehensive numerical example on the IEEE 39-bus New England test system, subject to two step-type disturbances.
\section{Background Material}
\subsection{Notation}
We consider a digraph $ \Gm = ( \Vm, \, \Em , A )$, where $ \Vm =\{ v_1, \, \ldots, v_n\}$ is a set of $n$ nodes and $\Em $ a set of $q$ edges over $ \Vm$, with the convention that $ (v_i , \, v_j) \in \Em $ means $ v_i \to v_j $, i.e., $ v_i $ is the tail and $ v_j $ is the head of the edge. The edge-induced adjacency matrix is $A=(a_{ij})_{(i,j)\in\Em}\in\mathbb{R}^{n\times n}$. For $A$ nonnegative and symmetric we define the Laplacian matrix $L\in\mathbb{R}^{n\times n}$, as $L=\text{diag}(\{\sum_{j=1}^na_{ij}\}_{i=1}^n)-A$, and the incidence matrix $\Delta\in\mathbb{R}^{n\times q}$, whose nonzero elements are defined component-wise as $\Delta_{ij} = 1$ if node $i$ is the sink node of edge $j$ and as $\Delta_{ij} = -1$ if node $i$ is the source node of edge $j$. It can be proven that $L=\Delta\,\diag(\{a_{ij}\}_{(i,j)\in\Em})\,\Delta^\top$ and that for a connected graph $\ker(\Delta^\top)=\ker(L)=\text{span}(\mathbf{1}_n)$, where $\mathbf{1}_n$ denotes the vector of ones of length $n$. A directed path $ \Pm$ from node $v_i$ to node $v_j$ is a sequence of distinct nodes $v_i, v_{i+1}, \dots , v_{j-1}, v_j$ s.t. $(v_k,v_{k+1})\in \mathcal{E}$ $\forall\,k\in\{i,i+1, \dots, j-1\}$: $\Pm = \{ v_i,  \dots , v_j \} $ (or $ \Pm^{v_i, v_j} $ when we need to specify the extremities of the path). Given a set of nodes $ \Zm \subseteq \Vm$, with a slight abuse of notation, we can use $ \Zm $ to express subspaces of $ \mathbb{R}^{n}$: $\Zm \sim \left\{  x \in \mathbb{R}^n  \; \text{s.t.} \; x_i \neq 0\;  \forall \; i \in \Zm  \; \text{and} \; x_i= 0\; \forall \; i  \notin \Zm \right\}$.
In the following, we explicitly fix a basis of canonical vectors for a vector subspace associated with a set of nodes. Let $ e_i $ be an elementary column vector (i.e., the vector having $1$ on the $i$-th slot and $0$ elsewhere); then given a set of nodes $\Zm\subseteq\Vm$, the vector subspace of $\mathbb{R}^n$ associated with $\Zm$ is $\textnormal{span}\{e_i\;\text{s.t.}\;v_i\in\Zm\}$.
For any $ \Zm \subset \Vm$, the notation $ \Zm^\perp$ is used for both the complement set $ \Vm \smallsetminus \Zm$ and the subspace orthogonal to (the vector space) $ \Zm $ in $ \mathbb{R}^n$. For each subspace $ \Zm$, $ \exists $ a full column rank matrix $ Z$ s.t. $ \im Z = \Zm $ and $ \ker Z =\{0\}$. Such matrix is called a basis matrix of $\Zm$. When $ \Zm $ is a set of nodes, then $ Z$ is a collection of elementary column vectors (i.e., one column for each node $ v_i \in \Zm$, having  1 in the $ v_i $-th slot and 0 elsewhere), or equivalently, $Z^\top$ is a collection of elementary row vectors. On $ \Gm = ( \Vm, \, \Em , A )$, for any subset of nodes $ \Wm\subseteq\Vm$ we can define the in- and out-boundary as follows. The in-boundary of $\Wm\subseteq\Vm$ in the digraph $ \Gm = ( \Vm, \, \Em , A )$ is the set of nodes in $\Wm$ that have at least one outgoing edge leading to a node in $\Wm^\perp$, i.e., $\partial_-(\Wm,A):=\{v_i\in\Wm\;\text{s.t.}\;\exists\,(v_i,v_j)\in\Em,\, v_j\in\Wm^\perp\}$. Similarly, the out-boundary of $\Wm$ is the set of nodes in $\Wm^\perp$ that have at least one incoming edge from $\Wm$, i.e., $\partial_+(\Wm,A):=\{v_j\in\Wm^\perp\;\text{s.t.}\;\exists\,(v_i,v_j)\in\Em,\, v_i\in\Wm\}$.
%For clarity, we note that three different adjacency matrices and their associated digraphs are used throughout the paper: $A$, $\tilde{A}$, and $\mathbf{A}$. %The matrix $A$ denotes the adjacency matrix describing how the oscillators are coupled in the nonlinear coupled-oscillator model. The matrix $\tilde{A}$ represents the coupling structure that arises when the model is linearized around a stable synchronized state. Finally, $\mathbf{A}$ denotes the state matrix in the associated descriptor system. 
\subsection{Coupled Oscillator Model}
Let us partition the node set as $\Vm=\Vm_1\cup\Vm_2$, with $r=|\Vm_1|$, $k=|\Vm_2|$ and $n=r+k$. Consider a symmetric, nonnegative edge-induced adjacency matrix $A$ with no self-loops. The Laplacian matrix $L\in\mathbb{R}^{n\times n}$ is symmetric and positive semi-definite. Each node in $\Vm$ is associated with a phase oscillator, whose dynamics is second-order Newtonian if the node is in $\Vm_1$ and first-order kinematic if the node is in $\Vm_2$, i.e.,
\begin{equation}\label{eq:nl_oscill}
\begin{aligned}
   M_i\ddot{\theta}_i+D_i\dot{\theta}_i &= {f}_i-\sum_{j=1}^n a_{ij}\sin(\theta_i-\theta_j),\qquad i\in\Vm_1,\\
   D_i\dot{\theta}_i &= {f}_i-\sum_{j=1}^n a_{ij}\sin(\theta_i-\theta_j),\qquad i\in\Vm_2,
\end{aligned}
\end{equation}
where $M_i>0$ is the inertia constant, $D_i>0$ is the damping coefficient, $f_i\in\mathbb{R}$ is the natural frequency, $\theta_i\in\mathbb{S}^1$ is the phase and $\dot{\theta}_i\in\mathbb{R}$ is the frequency of the oscillator $i$. For such dynamical system the state trajectories evolve in $\mathbb{T}^n\times\mathbb{R}^r$, where $\mathbb{T}^n=\mathbb{S}^1 \times \dots \times \mathbb{S}^1$ denotes the $n$-torus. If $i\in\Vm_2$, $\Vm_1=\emptyset$ with state $x=[\theta_1\dots\theta_k]^\top\in\mathbb{T}^k$, we obtain the Kuramoto model (\cite{dorfler2014synchronization}). If $i\in\Vm_1$, $\Vm_2=\emptyset$ with state $x=[\theta_1\dots\theta_r \,\dot{\theta}_1\dots\dot{\theta}_r]^\top\in\mathbb{T}^r\times\mathbb{R}^r$, we obtain the Mechanical Spring equations (\cite{dorfler2013synchronization}). We consider for generality both $\Vm_1$ and $\Vm_2$ to be nonempty. %The dynamics of the nonlinear system~\eqref{eq:nl_oscill} evolves on a torus, and thus it exhibits a wide range of phenomena (e.g., periodic orbits, stable and unstable equilibria, saddles, etc).

In this work, we focus on the case in which the coupled oscillator reaches a stable equilibrium in a rotating frame that rotates at the synchronization frequency. This case is called frequency synchronization (or phase-locking synchronization) in the literature (\cite{dorfler2014synchronization}). Moreover, we restrict ourselves to the case in which such synchronization is stable, i.e., the phases are cohesive, i.e., not only do they rotate at the same synchronization frequency $\omega^\ast$, but they are also pairwise bounded at a distance smaller than $\pi/2$. Formally, given $\theta^\ast=[\theta_1\dots\theta_n]^\top$ we have $\dot{\theta}^\ast=\omega^\ast\mathbf{1}_n$, and $|\theta_i-\theta_j|\leq\gamma$, $\gamma\in[0,\pi/2[$, $\forall\,(i,j)\in\Em$. The condition for obtaining a unique stable synchronization solution relies both on the network topology, through $L$, and on the natural frequencies difference for every connected pair of connected oscillators. If $f=[f_1\dots f_n]^\top$ denotes the vector of natural frequencies of each oscillator, and if we define the damping matrix as $\diag(\{D_i\}_{i\in\Vm})$, then the coupled oscillator model~\eqref{eq:nl_oscill} has a unique and stable solution $\theta^\ast$ with synchronized frequencies and cohesive phases if $||\Delta^\top L^\dagger\beta||_\infty\leq\sin(\gamma)$, being $L^\dagger$ the pseudoinverse of the network Laplacian matrix $L$, and $\beta=-\diag(\{D_i\}_{i\in\Vm})\omega^\ast\mathbf{1}_n+f$ the particular solution of the fixed-point equation $\beta=\Delta\,\diag(\{a_{ij}\}_{(i,j)\in\Em})\,\Delta^\top L^\dagger\beta$ (a rigorous treatment of this condition can be found in \cite{dorfler2013synchronization} and \cite{simpson2013synchronization}). %The condition for obtaining a unique stable synchronization solution relies both on the network topology, through $L$, and on the natural frequencies difference for every connected pair of connected oscillators (a rigorous treatment of this condition can be found in \cite{dorfler2013synchronization} and \cite{simpson2013synchronization}).
%If $f=[f_1\dots f_n]^\top$ denotes the vector of natural frequencies of each oscillator, and if we define the damping matrix as $\diag(\{D_i\}_{i\in\Vm})$, then the coupled oscillator model~\eqref{eq:nl_oscill} has a unique and stable solution $\theta^\ast$ with synchronized frequencies and cohesive phases if $||\Delta^\top L^\dagger\beta||_\infty\leq\sin(\gamma)$, being $L^\dagger$ the pseudoinverse of the network Laplacian matrix $L$, and $\beta=-\diag(\{D_i\}_{i\in\Vm})\omega^\ast\mathbf{1}_n+f$ the particular solution of the fixed-point equation $\beta=\Delta\,\diag(\{a_{ij}\}_{(i,j)\in\Em})\,\Delta^\top L^\dagger\beta$ (a rigorous treatment of this condition can be found in \cite{dorfler2013synchronization} and \cite{simpson2013synchronization}).
In this work, we implicitly assume that a stable synchronization solution exists, and we restrict our focus to the coupled oscillators that satisfy such a condition. Under such assumption, setting $\Vm_1=\{1,\dots,r\}$, $\Vm_2=\{r+1,\dots,r+k\}$, $\theta_{\Vm_1}=[\theta_1\dots\theta_r]^\top$ and $\theta_{\Vm_2}=[\theta_{r+1}\dots\theta_{r+k}]^\top$, and defining the state of the system as $x=[\dot{\theta}_{\Vm_1}^\top\;\theta_{\Vm_1}^\top\;\theta_{\Vm_2}^\top]^\top\in\mathbb{R}^N$, where $N=n+r$, it is possible to linearize the equations~\eqref{eq:nl_oscill} around $x^\ast=[\omega^\ast\mathbf{1}^\top_r\;\theta^{\ast\top}_{\Vm_1}\;\theta^{\ast\top}_{\Vm_2}]^\top$, $f^\ast=[f^{\ast\top}_{\Vm_1}\;f^{\ast\top}_{\Vm_2}]^\top=[f^\ast_1\dots f^\ast_{r}\;f^\ast_{r+1}\dots f^\ast_{r+k}]^\top$. Defining the deviations w.r.t. the linearization point as $\Tilde{\omega}_i=\dot{\theta}_i-\omega^\ast$, $\Tilde{\theta}_i=\theta_i-\theta_i^\ast$, and $\Tilde{f}_i=f_i-f_i^\ast$, we obtain the linearized coupled oscillator equations
\begin{equation}\label{eq:l_oscill}
\begin{aligned}
   M_i\dot{\Tilde{\omega}}_i&=-D_i\Tilde{\omega}_i-\sum_{j=1}^na_{ij}\cos(\theta^\ast_i-\theta^\ast_j)(\Tilde{\theta}_i-\Tilde{\theta}_j)+\Tilde{f}_i,\\
   \dot{\theta}_i&=\Tilde{\omega}_i,\quad i\in\Vm_1;\\
D_i\dot{\theta}_i&=-\sum_{j=1}^na_{ij}\cos(\theta^\ast_i-\theta^\ast_j)(\Tilde{\theta}_i-\Tilde{\theta}_j)+\Tilde{f}_i,\quad i\in\Vm_2.
\end{aligned}
\end{equation}

It is possible to define a new adjacency matrix $\Tilde{A}=(\Tilde{a}_{ij})_{(i,j)\in\Em}$\\$=(a_{ij}\cos(\theta^\ast_i-\theta^\ast_j))_{(i,j)\in\Em}\in\mathbb{R}^{n\times n}$ that is symmetric, nonnegative and does not have self-loops. Those properties are inherited from $A$ and by noting that the cosine is an even function and has image in $[0,1[$ since the phases are cohesive. It follows that the new Laplacian matrix $K=\text{diag}(\{\sum_{j=1}^n\Tilde{a}_{ij}\}_{i=1}^n)-\Tilde{A}=\begin{bmatrix}
    K_{a} & K_{b}\\
    K^\top_{b} & K_{c}
\end{bmatrix}\in\mathbb{R}^{n\times n}$, $K_a\in\mathbb{R}^{r\times r}$, $K_b\in\mathbb{R}^{r\times k}$, $K_c\in\mathbb{R}^{k\times k}$, 
is also symmetric and positive semi-definite. Defining the state and partitioning the natural frequency deviations as $\Tilde{x}=x-x^\ast$ and $\Tilde{f}=f-f^\ast=[\Tilde{f}_{\Vm_1}^\top\;\Tilde{f}_{\Vm_2}^\top]^\top$, and setting the mass-inertia matrix as $M = \text{diag}(\{M_i\}_{i\in\Vm_1})$, and the damping matrices as $D_{\Vm_1} = \text{diag}(\{D_i\}_{i\in\Vm_1})$, $D_{\Vm_2} = \text{diag}(\{D_i\}_{i\in\Vm_2})$, we obtain the following system in descriptor form
\begin{equation}\label{eq:descriptor}
    \mathbf{E}\dot{\Tilde{x}}=\mathbf{A}\Tilde{x}+\mathbf{B}\Tilde{f},
\end{equation}
with $\mathbf{E},\,\mathbf{A}\in\mathbb{R}^{N\times N}$ and $\mathbf{B}\in\mathbb{R}^{N\times n}$ of the form
\begin{equation*}
    \mathbf{E}\hspace{-1pt}=\hspace{-2pt}\begin{bmatrix}
        M&\mathbf{0}&\mathbf{0}\\
        \mathbf{0}&I_r&\mathbf{0}\\
        \mathbf{0}&\mathbf{0}&D_{\Vm_2}
    \end{bmatrix}\hspace{-3pt},\,
    \mathbf{A}\hspace{-1pt}=\hspace{-2pt}\begin{bmatrix}
        -D_{\Vm_1}&-K_a&-K_b\\
        I_r&\mathbf{0}&\mathbf{0}\\
        \mathbf{0}&-K^\top_b&-K_c
    \end{bmatrix}\hspace{-3pt}, \,
    \mathbf{B}\hspace{-1pt}=\hspace{-2pt}\begin{bmatrix}
        I_r&\mathbf{0}\\
        \mathbf{0}& \mathbf{0}\\
        \mathbf{0}&I_k
        \end{bmatrix}\hspace{-3pt},
\end{equation*}
with $N=2r+k=n+r$, $I_\alpha$ denoting the identity matrix of dimension $\alpha$ and $\mathbf{0}$ a zero matrix whose dimensions have been omitted for clarity. We assume hereafter $\mathbf{E}$ to be non-singular, which automatically guarantees that the pair $(\mathbf{E},\mathbf{A})$ is regular, and the input signal $\Tilde{f}$ to be smooth and sufficiently small to not compromise stability. For the system~\eqref{eq:descriptor} it is possible to define an extended graph $\overline{\Gm}=(\overline{\Vm},\overline{\Em},\mathbf{A})$ with adjacency matrix $\mathbf{A}$.

\subsection{Electromechanical Power System Model}\label{sec:electr_model}
Eq.~\eqref{eq:nl_oscill} can be used to describe a simplified power system model which is typically used for the control of the frequency in a purely inductive transmission network. The transmission interconnection is governed by $a_{ij}$\footnote{$a_{ij}=\frac{V_i V_j}{X_{ij}}$ where $V_i$ and $V_j$ are the voltage magnitudes in nodes $i$ and $j$. $X_{ij}$ is the line reactance of a line connecting nodes $i$ and $j$. If there is no connection between nodes $i$ and $j$, $a_{ij}=0$. }. 
In power system modeling, the part associated with $\Vm_1$ is often called the swing equation, and it is used to model the first-order dynamics of the mechanical swing of the generators, with inertia constant $M_i$, in a high voltage transmission grid (where lines can be approximated as purely inductive and voltages are assumed constant) (\cite{Glover}). The part in \eqref{eq:nl_oscill} associated with $\Vm_2$ is related to the buses in the power system without a generator. In the linearized case in \eqref{eq:l_oscill}, $\Tilde{\theta}$ and $\Tilde{\omega}$ are the phase and frequency deviations from the nominal frequency of the grid. The damping matrix $D_{\Vm_{2}}$ is very small, often approximated with the zero matrix, but for our numerical examples, we will approximate it with $\epsilon I_k$, with $\epsilon=10^{-4}$, to avoid a singular $\mathbf{E}$. 
The choice of $\epsilon$ arises from the phase velocity $v_p=1/\sqrt{L'C'}$ \cite[p. 10]{pai89}. This gives $t_p$ between $0.1-1.2\,ms$ as the propagation time for the transmission lines of the numerical example, in Section~\ref{NumEx}\footnote{$L'$ is the inductance per unit (p.u.) length, $C'$ is the capacitance p.u. length. Propagation time $t_p$ is calculated by $t_p=l\sqrt{L'C'}$, with line length $l$ given in the p.u. length.} \cite[Appendix A]{pai89}. Since many lines enter a node, we choose the smallest, making $\epsilon=10^{-4}\,s$. %This corresponds to replacing a very short time delay $\epsilon$ with a low-pass filter with time constant $\epsilon$.
In power systems, something called droop control is often used to achieve frequency control. This is a control loop that changes active power output from generators in response to locally measured frequency deviations, to stabilize the frequency \cite[p.28-33]{Mach20}. Denote the static droop constant $e_p$, which is often defined in a \% term, usually around 5\%. This means that the frequency is allowed to deviate 5\% of the generator power deviation from nominal. %  In control terms this translates to a P-controller with gain 20 \rednote{LL: (why is it 20? Because 1/0.05=20. The nominal power is not defined numerically.)}. Define the constant $k_{ij}:=P_{\mathrm{max},ij} \cos(\varphi_i^0 - \varphi_j^0) $ for simplicity, and note that $k_{ij}=k_{ji}\geq 0$. Call the nominal power of generator $i$ $P_i$. \eqref{eq:LinSwing} can be written in matrix form and with static droop control this becomes
From this, we get a simple expression for the deviation in mechanical power applied in \eqref{eq:l_oscill} with $D_i=- \frac{1}{e_p}\frac{P_i}{\omega^*}$, where $P_i$ is the rated power of generator $i\in \Vm_1$. There are also damper windings in the generator that add to this damping term, but they are much smaller than the damping given by the control \cite[p.195-199,209-210]{Mach20}. 

\subsection{Minimal Input Cardinality Disturbance Decoupling}
In system~\eqref{eq:descriptor}, we note that the matrix $\mathbf{B}$ has $n$ elementary columns; therefore, the range of $\mathbf{B}$ identifies the set of all admissible natural frequency deviations (interpreted as control inputs or disturbances) that may enter the system. The set of admissible input nodes is thus $\im(\mathbf{B})\cong\Vm_{\mathbf{B}}\subset\overline{\Vm}$. We assume that only a subset of the admissible natural frequency deviations is active in system~\eqref{eq:descriptor}. We denote with $\Vm_{\Tilde{f}}\subseteq\Vm_{\mathbf{B}}$, where $|\Vm_{\Tilde{f}}|=\ell\leq n$, the index set of the $\ell$ nonzero entries of the input $\Tilde{f}\in\mathbb{R}^n$ entering in the system~\eqref{eq:descriptor}. We denote with $\Dm\subseteq\Vm_{\Tilde{f}}$, where $|\Dm|=d$, the set of disturbances, and with $\Imin\subseteq\Vm_{\Tilde{f}}$, where $|\Imin|=m$, the set of control inputs. Assuming $\Dm\cap\Imin=\emptyset$, and setting $\Dm\cup\Imin=\Vm_{\Tilde{f}}$, we can therefore split the input $\mathbf{B}\Tilde{f}$ into the sum of an exogenous disturbance vector $Dw$ and a control input vector $Bu$. Here the matrices $D\in\mathbb{R}^{N\times d}$ and $B\in\mathbb{R}^{N\times m}$ have elementary columns and satisfy $\im(D)\cong\Dm$, $\im(B)\cong\Imin$, with $d+m=\ell$. We denote with $\Cm$, where $|\Cm|=p\leq N$, the set of output measurements and with $\Tm$, where $|\Tm|=t\leq N$, the set of targets. Both $ \Cm$ and $ \Tm $ correspond to sets of nodes: defining the matrices $C\in\mathbb{R}^{m\times N}$ and $T\in\mathbb{R}^{t\times N}$ with elementary rows and s.t. $\im(C^\top)\cong\Cm\subseteq\overline{\Vm}$ and $\im(T^\top)\cong\Tm\subseteq\overline{\Vm}$, and assuming $\Cm\cap\Tm=\emptyset$, we obtain the following system
\begin{equation}
\begin{split}
\mathbf{E}\dot{\Tilde{x}} & =  \mathbf{A}\Tilde{x} + B u + D w, \\
y & = C \Tilde{x}, \\
z & = T\Tilde{x}. 
\end{split}
\label{eq:lin-syst_ddp}
\end{equation}
For this system, the classical disturbance decoupling problem can be formulated as follows.

\begin{prob}[\textbf{DDP}]\label{problem:DDP0}
Given the system~\eqref{eq:lin-syst_ddp}, find a control law $u(t)$ that renders the target nodes in $\mathcal{T}$ unaffected by the disturbances in $\mathcal{D}$.
\end{prob}

It is well known that the solvability conditions for the DDP depend on the control law that is employed, and are geometrically formulated exploiting set inclusions involving controlled and conditioned invariant sets. As is proven in Propositions 2 and 3 of \cite{lebon2026geometric_tac}, under the setting of the paper the characterization of controlled and conditioned invariance can be carried out in terms of node sets. Since the matrix $\mathbf{E}$ is square, diagonal, and non-singular, the same definitions apply to the system~\eqref{eq:lin-syst_ddp} in the descriptor form. We give the following definitions of controlled and conditioned invariant sets of nodes for system~\eqref{eq:lin-syst_ddp}.
\begin{defn}\label{def:controlled_inv}
$ \Zm \subseteq \overline{\Vm}$ is a controlled invariant set of nodes if $\forall\,(v_i , \, v_j) \in \overline{\Em} $, $ v_i \in \Zm \; \Longrightarrow \; v_j \in \Zm \cup \Imin$.
\end{defn}
\begin{defn}\label{def:conditioned_inv}
$\Sm\subseteq\overline{\Vm} $ is a conditioned invariant set of nodes if $\forall\,(v_i , \, v_j) \in \overline{\Em} $, $ v_i \in \Sm$, $ v_i \notin \Cm$ $ \Longrightarrow \; v_j \in \Sm$.
\end{defn}

Applying the iterative schemes of Propositions~4 and 5 of \cite{lebon2026geometric_tac}, it is always possible to construct, resp., the maximal controlled invariant subset of a given node set $\Zm_0$, and the minimal conditioned invariant superset of a given node set $\Sm_0$. Assuming the initializations $\Zm_0=\overline{\Vm}\smallsetminus\Tm$ and $\Sm_0=\Dm$, we denote with $\Zm^\circ(\Imin)$ and $\Sm^\circ(\Cm)$ resp., the maximal controlled invariant node set contained in $\overline{\Vm}\smallsetminus\Tm$ and the minimal conditioned invariant node set containing $\Dm$. The node sets $\Zm^\circ(\Imin)$ and $\Sm^\circ(\Cm)$ are closely related to the maximal controlled invariant and the minimal conditioned invariant subspaces of the classical geometric control literature, denoted by $\Zm^\ast(\Imin)$ and $\Sm^\ast(\Cm)$ resp., that are computed by the standard recursions of \cite{basile1969controlled} and \cite{wonham1970decoupling}. In particular, assuming that the recursions start from the same initial conditions, in Theorem~1 of \cite{lebon2026geometric_tac} it is proven that $\Zm^\circ(\Imin)\subseteq\Zm^\ast(\Imin)$ and $\Sm^\ast(\Cm)\subseteq\Sm^\circ(\Cm)$. Therefore, since from Theorem~1.1 of \cite{fletcher1989disturbance} the standard disturbance decoupling condition by state feedback $\Dm\subseteq\Zm^\ast(\Imin)$ also applies to the system~\eqref{eq:lin-syst_ddp}, it follows that a sufficient condition for the solvability of the disturbance decoupling problem for the system~\eqref{eq:lin-syst_ddp} is given by the inclusion of node sets $\Dm\subseteq\Zm^\circ (\Imin)$. Selecting the disturbances and the targets s.t. $\Dm\cap\Tm=\emptyset$, we define the \emph{minimal input cardinality disturbance decoupling problem} as
\begin{equation}
\qquad \Biggl\{
    \begin{aligned}
      & \min_{\Imin} \quad |\Imin|\\
      & \text{subject to:} \quad \Dm\subseteq\Zm^\circ (\Imin)
    .\end{aligned}\label{eq:opt_problem1}
\end{equation}

As has been thoroughly investigated in \cite{lebon2026geometric_tac}, the optimal solution to Problem~\eqref{eq:opt_problem1} can be found in polynomial time and corresponds to the out-border of the set $\Zm^\circ (\Imin)$, i.e., $\Imin=\partial_+(\Zm^\circ(\Imin),\mathbf{A})$ (the proof can be found in Theorem 6 of \cite{lebon2026geometric_tac}). A solution is admissible for the system~\eqref{eq:descriptor} if $\Imin\subseteq\Vm_{\mathbf{B}}$. Since this work focuses solely on admissible solutions, the term admissible will be implied whenever a solution is mentioned.

\subsubsection{Disturbances and control inputs in the power system model}
In the power system model, disturbances of active power $w$ can be added to both generator and non-generator nodes. In the generator case, it can stem from a fault at the generator or an uneven supply of primary mechanical power, affecting the power output. At a load the power consumption can deviate from the linearization point, due to stochasticity in the electric power consumption of customers. Larger disturbances can also happen at a load due to a sudden disconnection of an area of consumption (\cite{Bryne14, SVK17}). We can choose to add or withdraw active power $u$, in a control effort, at both generators and non-generator buses. At generator buses, this can be added from the mechanical power applied to the generator, while at the non-generator buses, this can be done by actively changing consumption, or installing a battery that can both add or withdraw power to that node. %This will be our control input. 

\section{Main Result}
%\rednote{I would add the following sentence to justify DDPOF: The problem~\eqref{eq:opt_problem1} is formulated in terms of state feedback. This is not very realistic in power grid applications, as it requires to measure of phases at all nodes of $ \overline{\Vm}$. A more parsimonious scheme is to consider an output feedback law, which relies only on the measured outputs in the set $ \Cm$. }
In this section, the solution of the minimal input cardinality disturbance decoupling Problem~\eqref{eq:opt_problem1} is exploited for the design of an output feedback that achieves the decoupling. In fact, we note that Problem~\eqref{eq:opt_problem1} is formulated in terms of state feedback. This is not very realistic in power grid applications, as it requires measuring the states at all nodes of $ \overline{\Vm}$. A more parsimonious scheme is to consider an output feedback law, which relies only on the measured outputs in the set $ \Cm$. In the literature, this is known as the disturbance decoupling problem via output feedback (DDPOF). The choice of proposing an output feedback law $u=-Gy$ for the control input entering in system~\eqref{eq:lin-syst_ddp} is driven primarily by three motivations: 1) Reducing the norm of the control input $Bu$ entering in the state equation~\eqref{eq:descriptor} by minimizing the impact of the feedback on the preexisting topology in the closed-loop. 2) Reducing the set of measurements from $\overline{\Vm}$ (corresponding to the state feedback solution) to $\Cm$. 3) Reducing the delay in the sensing-actuation protocol (a dynamical feedback solution would have minimized the measurement set but augmented the computational time to provide the state estimate). The graphical characterization of the DDPOF over networks can be found in Theorem~3 of \cite{lebon2026geometric_tac}.
%\begin{thm}[Theorem 3 of \cite{lebon2026geometric}]\label{thm:ddpof}
%Problem~\ref{problem:DDP0} is solvable for the system~\eqref{eq:lin-syst_ddp} via the output feedback law $u=-Gy$ if any of the following equivalent conditions is met:
%\benu[label=(\alph*)]
%\item $\exists\,\Wm$ controlled invariant and conditioned invariant set of nodes s.t. $\Dm\subseteq\Wm\subseteq\Vm\smallsetminus\Tm$;\label{item:a_thm:ddpof}
%\item  Let $\Pm_1,\dots,\Pm_\ell$ be all possible $\Dm$-to-$\Tm$ paths in $\overline{\Gm}$. For every path $\Pm_j=\{v_1,\dots,v_{\kappa_j}\}$ of length $\kappa_j-1\geq1$, with $v_1\in\Dm$ and $v_{\kappa_j}\in\Tm$, $\exists$ at least one sub-path of length 1 of the form $\{v_{p},v_{p+1}\}$ with $v_{p}\in\Cm$ and $v_{p+1}\in\Imin$, $p\in\{1,\dots,\kappa_j-1\}$, $j\in\{1,\dots,\ell\}$.\label{item:c_thm:ddpof}
%\eenu
%\end{thm}
In the following, the output feedback that solves the DDPOF from the optimal solution $\Imin$ of the minimal input cardinality disturbance decoupling Problem~\eqref{eq:opt_problem1} is derived by first constructing a controlled and conditioned invariant node set $\Wm$ from $\Imin$ and $\Cm$.
Let us denote by $ \mathcal{P} = \bigcup_i \mathcal{P}_i $ the union of all $ \Dm$-to-$ \Tm$ paths in $\overline{\Gm}$. 
\begin{lem}\label{lem:1}
Given an optimal solution $\Imin=\partial_+(\Zm^\circ(\Imin),\mathbf{A})$ of Problem~\eqref{eq:opt_problem1}, the node set $\Wm:=\Zm^\circ(\Imin)$ is s.t. $\Dm\subseteq\Wm\subseteq\overline{\Vm}\smallsetminus\Tm$ and $\Imin=\partial_+(\Wm,\mathbf{A})=\partial_+(\Wm,\mathbf{A})\cap\Pm$.
\end{lem}
\begin{pf}
$\Dm\subseteq\Wm$ follows from the fact that $\Dm\subseteq\Pm$ and $\Dm\subseteq\Zm^\circ(\Imin)$, since $\Imin$ is a solution of Problem~\eqref{eq:opt_problem1}. $\Wm\subseteq\overline{\Vm}\smallsetminus\Tm$ follows by noting that $\Wm$ is a subset of $\Zm^\circ(\Imin)$, which is contained in $\overline{\Vm}\smallsetminus\Tm$. The fact that $\partial_+(\Wm,\mathbf{A})=\partial_+(\Wm,\mathbf{A})\cap\Pm$ comes directly from item $(b)$ of Theorem~6 in \cite{lebon2026geometric_tac}.
\end{pf}
\begin{lem}\label{lem:2}
   Given an optimal solution $\Imin=\partial_+(\Zm^\circ(\Imin),\mathbf{A})$ of Problem~\eqref{eq:opt_problem1}, and the set $\Wm:=\Zm^\circ(\Imin)$, if $\Cm=\partial_-(\Wm,\mathbf{A})$, then $\Wm$ is a controlled and conditioned invariant set of nodes.
\end{lem}
\begin{pf}
From Lemma 4 in the Appendix C of \cite{lebon2026geometric_tac}, if $\Wm$ is a conditioned and controlled invariant set of nodes, then $\Imin\supseteq\partial_+(\Wm,\mathbf{A})$ and $\Cm\supseteq\partial_-(\Wm,\mathbf{A})$. The vice-versa is also true, in fact, if we assume by contradiction that $\Wm$ is not a controlled invariant set of nodes, from Definition~\ref{def:controlled_inv} $\exists$ at least a node $v_i\in\Wm$ and an edge $(v_i,v_j)\in\overline{\Em}$ s.t. $v_j\notin\Wm\cup\Imin$, which is a contradiction since $\Imin\supseteq\partial_+(\Wm,\mathbf{A})$. Similarly, if we assume by contradiction that $\Wm$ is not a conditioned invariant set of nodes, from Definition~\ref{def:conditioned_inv} $\exists$ at least a node $v_i\in\Wm\smallsetminus\Cm$ and an edge $(v_i,v_j)\in\overline{\Em}$ s.t. $v_j\notin\Wm$, which is a contradiction since $\Cm\supseteq\partial_-(\Wm,\mathbf{A})$. Noting that the inclusions are not strict, the statement follows from Lemma~\ref{lem:1}.
\end{pf}
\begin{thm}
   Given an optimal solution $\Imin=\partial_+(\Zm^\circ(\Imin),\mathbf{A})$ of Problem~\eqref{eq:opt_problem1}, the node sets $\Wm:=\Zm^\circ(\Imin)$ and $\Cm=\partial_-(\Wm,\mathbf{A})$, the unique output feedback $G\in\mathbb{R}^{m\times p}$, called ``friend'' of $\Wm$, i.e., s.t. $(\mathbf{A}-BGC)\Wm\subseteq\Wm$, is given by $G=B^\top\mathbf{A}C^\top$. 
\end{thm}
\begin{pf}
The proof follows directly from Lemma~\ref{lem:1} and Lemma~\ref{lem:2}, and from Theorem~11 of \cite{lebon2026geometric_tac}.
\end{pf}
%\subsection{Interpretation}
%Graphically, the effect of $G=B^\top \mathbf{A}C^\top$ is to remove the sub-paths of %length 1 described in item~$\ref{item:c_thm:ddpof}$ of Theorem~\ref{thm:ddpof}, that is, %canceling in the rows of $\mathbf{A}$ identified by $\Imin$, the elements in the columns %identified by $\Cm$.

\subsection{Practical Considerations}\label{Sec:PracCon}

For a power system application, the disturbance decoupling control strategy requires that the phases of some of the nodes are measured. This can be done using so-called Phasor Measuring Units (PMU), which have to be installed at the nodes in the set of output measurements $\Cm$, and communicated to the set of control inputs $\Imin$. In the upcoming example, we will apply the disturbance decoupling to the simplified power system model described in Section~\ref{sec:electr_model}. This is a heterogeneous system with rather slow dynamics in $\Vm_1$ and very fast dynamics in $\Vm_2$, rendering the system almost singular. This will result in an almost ideal step-like control action from the active power injection under some of the possible disturbances. The control actuation is performed by a battery, or possibly through load control, and cannot act instantaneously without risk of damage to the components. To account for this and to incorporate communication time from the PMUs, the applied control is the ideal signal filtered through a low-pass filter with different time constants. Here we lump both time delay and the desired smoothening of the control input into a single low-pass filter with different time constants. This shows that the disturbance decoupling works well also under these practical considerations. 
%\rednote{LL: this part is the critical one, since it is the one applied to the almost singular power network case. Battery, inductance compensation, load control, will be here.}
\section{Numerical Example}\label{NumEx}
In this numerical section, we apply our results to the IEEE 39-bus system with 10 generators, first developed by \cite{Ath79}. The system represents the New England electric power transmission grid, in the north-east of the U.S.A.. This choice is motivated by the fact that the considered power network is an IEEE benchmark system (\cite{canizares2016benchmark}), and that for this network, a stable phase-locking synchronization exists (see Table 1 in \cite{dorfler2013synchronization}), and the corresponding linearized system is \emph{small-signal stable} (\cite{datta2020small}). For this system, the nominal frequency is $60$ Hz and the p.u. data is taken from \cite[Appendix A]{pai89} to build $\mathbf{E}$ and $\mathbf{A}$ as in the system~\eqref{eq:lin-syst_ddp}. The generator nodes (numbered 30 to 39 in \cite[Appendix A (p. 224)]{pai89}, whose dynamics are second-order Newtonian, are nodes in $\Vm_1$. Similarly, the non-generator nodes, whose dynamics are first-order kinematic, are nodes in $\Vm_2$. As the numbering of the network nodes of this power network is established within the power systems community, to whom this paper is partially addressed, the original ordering is retained. The system corresponds to the state ordering in equation~\eqref{eq:descriptor} if the generator buses are assigned to the first $10+10$ components of the state vector, representing resp. the generators’ frequencies and phases.

%\begin{example}
\emph{Problem Setup:}
Once $\mathbf{A}$ is constructed, it is possible to associate the extended digraph $\overline{\Gm}$, as shown in Fig.~\ref{fig:ex1_2}. For this test system, we consider $\epsilon=10^{-4}$ for the construction of $D_{\Vm_2}$ in $\mathbf{E}$, as discussed in Section~\ref{sec:electr_model}. We take our disturbance and target sets as resp. $\Dm=\{22,\,44\}$ and $\Tm=\{40,\,41\}$.%\\\\

\emph{Solution:} The minimal input cardinality solution of Problem~\eqref{eq:opt_problem1} is $\Imin=\{16\}$. Following the procedure outlined in Lemmas 1 and 2, with $\Imin=\{16\}$ the disturbance decoupling via output feedback is solvable if $\Cm=\{19,\,21,\,24\}$. The decoupling friend $G=[55.3272\, \; 78.7903\, \;181.4040]$$=e_{16}^\top \mathbf{A}[e_{19}\,e_{21}\,e_{24}]$ has the effect of canceling the edges $(19,16)$, $(21,16)$ and $(24,16)$ in closed-loop. 

\emph{Simulations:}
To verify our results, we assume that a step of power of amplitude 1 (corresponding to the p.u. base of 100 MW) is active from $t=0$ in the disturbance node $44$, directly affecting the frequency deviation of generator $44$, while a step of power of amplitude 0.5 (corresponding to 50 MW in the p.u. base) is active from $t=20\,s$ in the disturbance node $22$, directly affecting the phase deviation of node $22$. We make the choice of step disturbances since it can clearly be illustrated in the results, and it is the common disturbance shape when considering the dimensioning disturbance for system design (\cite{SVK17}). This corresponds to sudden disconnection of a generator, or consumption cluster. Since we work with a linearized model, the sign of the disturbance only affects the sign of the output. To be able to see the combined effect of the two disturbances we choose both to have positive sign, although this would in reality not be the case for a disturbance at generator node $44$. We denote with $\Vm_{\text{dec}}\subseteq\overline{\Vm}$ the subset of network nodes that are decoupled in the closed-loop (in Fig.~\ref{fig:ex1_2} such nodes correspond to the nodes to the left of node $16$) and with $\Vm_{\text{dist}}\subseteq\overline{\Vm}$ the subset of network nodes that are disturbed in the closed-loop (such nodes correspond to the nodes to the right of node $16$ in Fig.~\ref{fig:ex1_2}). Clearly $\Tm\subseteq\Vm_{\text{dec}}$ and $\Dm\subseteq\Vm_{\text{dist}}$. As shown in Fig.~\ref{fig:ex1_3}, the decoupling output feedback applied at node~16 not only counteracts the propagation of the disturbance toward $\Vm_{\text{dec}}$, but also asymptotically stabilizes the frequency deviations of the generators in $\Vm_{\text{dist}}$. In the open-loop case, these frequency deviations converge to a nonzero steady-state value, leading to a linear drift of the nodes’ phase deviations in $\overline{\Vm}$ at steady-state, characterized by identical slopes. In the time interval $[0,20\,s[$ we observe that the phase deviations grow inversely proportional to the electrical distance from the disturbed node $44$. Around $t=20\,s$, the nodes’ phase deviations exhibit a sudden jump, whose amplitude is inversely proportional to the electrical distance from the disturbed node~$22$. The initial conditions are set to zero, and the simulation spans one minute.

\emph{Practical Considerations:} 
We aim to verify whether the proposed methods remain valid when the sensing-actuation protocol is not instantaneous, thereby representing a more realistic scenario, as described in Section~\ref{Sec:PracCon}. To do so, we assume the system being controlled in closed-loop by a filtered control input $u_{\text{real}}$, which in the Laplace domain is given by the product of a low-pass filter $\frac{1}{\tau s+1}$ and the exact output feedback $u_{\text{ideal}}$. We simulate the impact on the frequency deviations of the generators in $\Vm_{\text{dec}}$ for different values of the time constant $\tau$. As shown in Fig.~\ref{fig:ex1_4}, even with a time constant of $1\,s$, the frequency deviation impulses generated by the disturbance at $t=20\,s$ exhibit a peak of approximately $0.005$ Hz, demonstrating the robustness of the proposed decoupling strategies with respect to sensing-actuation delays. The filtered control input is plotted for different time constants and is compared with the exact output feedback. 

\emph{Discussion:}
In Fig.~\ref{fig:ex1_4} it can be seen that the steady state control input becomes the sum of the 2 disturbances, i.e., the control input completely compensates for the disturbances. 
%One of the reasons for static droop is that there becomes a steady state frequency deviation that a secondary control infrastructure can act on, after the immediate, automatic response. Who is responsible for this is often determined by a market mechanism. \rednote{add ref if we keep these sentences} However, with our approach, this frequency deviation is completely eliminated. Instead of a steady state frequency deviation for the market to act on, it can instead act on the power output from our controller, if it is sent to the transmission system operator. 
The approach presented in this paper results in a feedback from phase measurements to the power applied in node 16. This requires measurement of the aforementioned phases through a PMU. The effect of the controller is that it compensates for the loss of active power in one region by inserting it in the node(s) of choice. The same result can thus be achieved by measuring the power flow entering the node(s) in $\Imin$ from the disturbed area and adding the lacking active power. This would result in a local controller with similar behavior. Different values of $\epsilon$ and $\tau$ were also tested. A smaller $\epsilon$ made no noticeable difference. If both $\epsilon$ and $\tau$ are really small ($<10^{-6}$), the control action becomes a step of $-50$~MW at time $t=20\,s$, with even smaller frequency deviations.
\begin{figure}[htb!]
    \centering
    \includegraphics[clip=true, width=\linewidth]{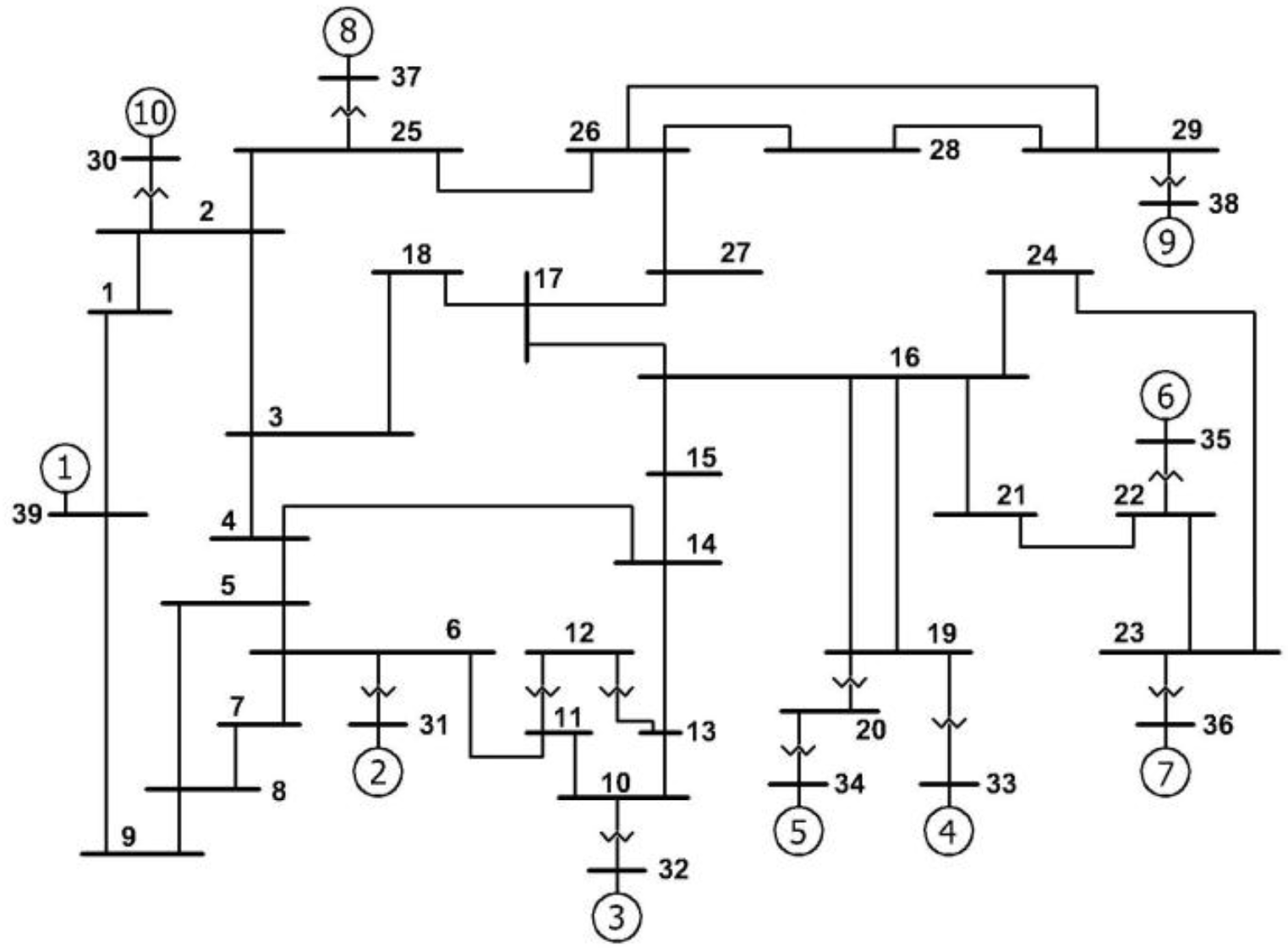}
    \caption{IEEE New England 39 Bus test system single-line diagram \cite[Appendix A (p. 224)]{pai89}. }
    \label{fig:ex1_1}
\end{figure}
\begin{figure}[htb!]
    \centering
    \includegraphics[trim=7.5cm 7cm 7cm 6cm, clip=true, width=1\linewidth]{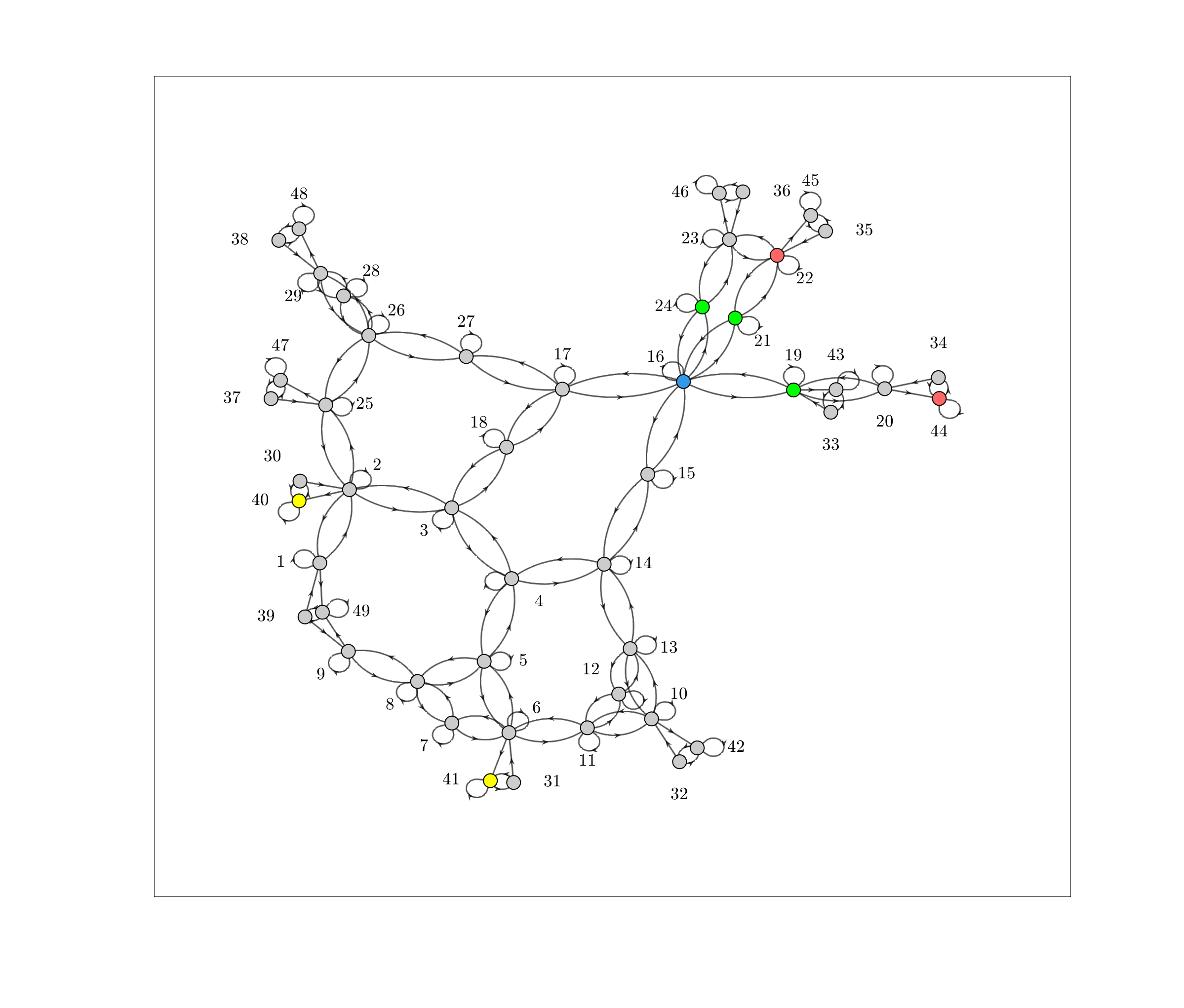}%\vspace{-0.2cm}
    \caption{Digraph of IEEE New England 39 Bus. Disturbances are in red, targets in yellow, control inputs in blue, output nodes in green.}
    \label{fig:ex1_2}
\end{figure}
\begin{figure}[htb!]
    \centering
    \includegraphics[clip=true, width=1\linewidth]{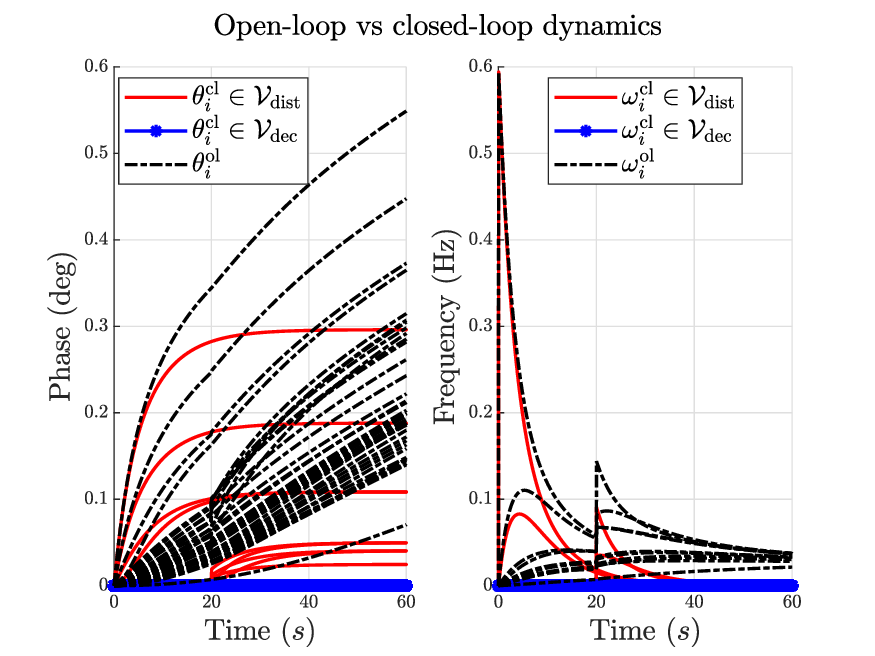}\vspace{-0.2cm}
    \caption{Decoupling effect in closed-loop of the output feedback over $\Vm_{\text{dec}}$ and its stabilizing effect over $\Vm_{\text{dist}}$. The superscripts ``ol" and ``cl" stand for ``open loop" and ``closed loop", respectively.}
    \label{fig:ex1_3}
\end{figure}
\begin{figure}[htb!]
    \centering
    \includegraphics[clip=true, width=1.1\linewidth]{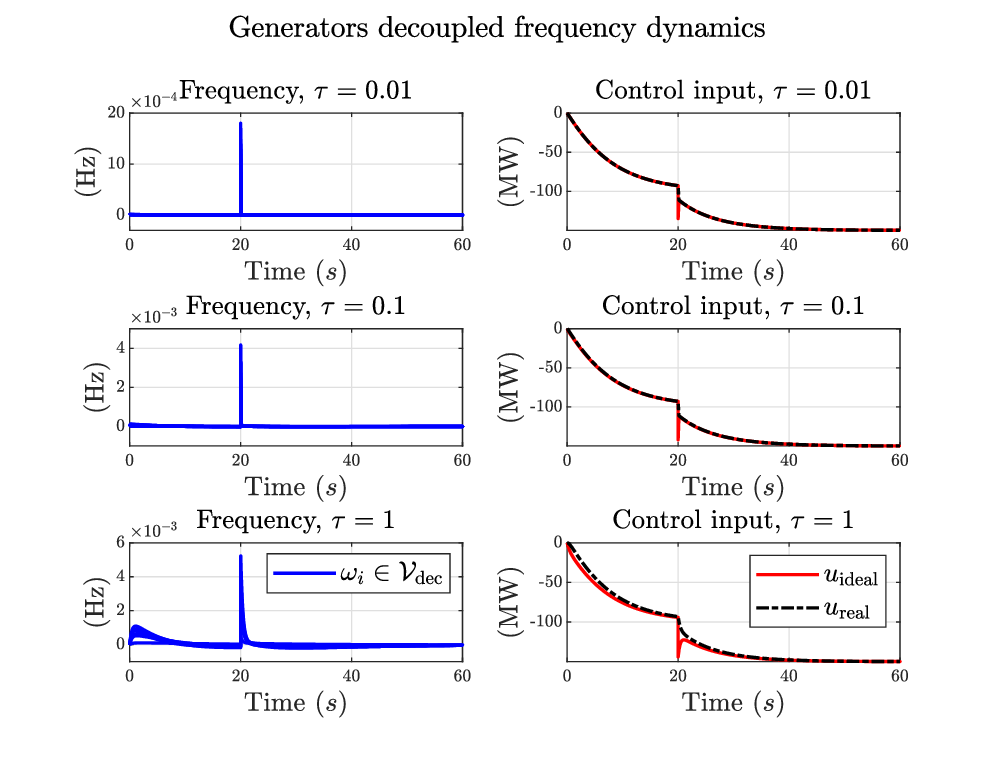}\vspace{-0.2cm}
    \caption{Decoupling of generators' frequency deviations in closed-loop with delayed output feedback control.}
    \label{fig:ex1_4}
\end{figure}
\label{ex:ex1}
%however, it would not necessarily drive the phases in the target region to zero, as seen in \ref{fig:ex1_3}. The reason for this is that without a direct phase reference from the PMUs, there is nothing that compensates for the transient effects of the disturbance when the control action is applied through the low-pass filter. The phases in the target area will all shift with the same value, and all frequency deviations will be zero, since this constitutes identical power flows to the situation before the disturbances. 
%\hfill\qed
%\end{example}
\section{Conclusions}
We have examined the problem of minimal input cardinality disturbance decoupling via output feedback in a system of coupled oscillators, linearized around a stable phase-locked synchronized state. When applied to power networks, this method ensures the real-time decoupling of disturbances originating in one region from propagating toward a designated region to be protected. This is achieved without interrupting nominal operation, %, i.e., without disconnecting transmission lines or isolating generators affected by the disturbance,
thereby preserving the continuity and stability of the power supply. The numerical simulations have been carried out on the IEEE 39-bus New England power network. These results substantiate the efficacy of the decoupling mechanism and its resilience to sensing-actuation delays, while ensuring full compatibility with the internal stability properties of the system. Future work should extend the results presented in this paper to more detailed power network models that account for voltage dynamics, reactive power flows, and higher-order generator dynamics.

%\bibliography{ifacconf}

\begin{thebibliography}{20}
\providecommand{\natexlab}[1]{#1}
\providecommand{\url}[1]{\texttt{#1}}
\providecommand{\urlprefix}{URL }
\expandafter\ifx\csname urlstyle\endcsname\relax
  \providecommand{\doi}[1]{doi:\discretionary{}{}{}#1}\else
  \providecommand{\doi}{doi:\discretionary{}{}{}\begingroup
  \urlstyle{rm}\Url}\fi

\bibitem[{Athay et~al.(1979)Athay, Podmore, and Virmani}]{Ath79}
Athay, T., Podmore, R., and Virmani, S. (1979).
\newblock A practical method for the direct analysis of transient stability.
\newblock \emph{IEEE Trans. Pow. Appar. and Sys.}, PAS-98(2), 573--584.
\newblock \doi{10.1109/TPAS.1979.319407}.

\bibitem[{Basile and Marro(1969)}]{basile1969controlled}
Basile, G. and Marro, G. (1969).
\newblock Controlled and conditioned invariant subspaces in linear system
  theory.
\newblock \emph{J. of Opti. Theo. and Appl.}, 3, 306--315.

\bibitem[{Byrne et~al.(2014)Byrne, Trudnowski, Neely et~al.}]{Bryne14}
Byrne, R.H., Trudnowski, D.J., Neely, J.C., et~al. (2014).
\newblock Optimal locations for energy storage damping systems in the western
  north american interconnect.
\newblock In \emph{2014 IEEE PES Gen. Meet. | Conf. \& Expo.}, 1--5.
\newblock \doi{10.1109/PESGM.2014.6939875}.

\bibitem[{Canizares et~al.(2016)Canizares, Fernandes, Geraldi
  et~al.}]{canizares2016benchmark}
Canizares, C., Fernandes, T., Geraldi, E., et~al. (2016).
\newblock Benchmark models for the analysis and control of small-signal
  oscillatory dynamics in power systems.
\newblock \emph{IEEE Trans. Pow. Sys.}, 32(1), 715--722.

\bibitem[{Datta(2020)}]{datta2020small}
Datta, S. (2020).
\newblock Small signal stability criteria for descriptor form power network
  model.
\newblock \emph{IJC}, 93(8), 1817--1825.

\bibitem[{D{\"o}rfler and Bullo(2014)}]{dorfler2014synchronization}
D{\"o}rfler, F. and Bullo, F. (2014).
\newblock Synchronization in complex networks of phase oscillators: A survey.
\newblock \emph{Automatica}, 50(6), 1539--1564.

\bibitem[{D{\"o}rfler et~al.(2013)D{\"o}rfler, Chertkov, and
  Bullo}]{dorfler2013synchronization}
D{\"o}rfler, F., Chertkov, M., and Bullo, F. (2013).
\newblock Synchronization in complex oscillator networks and smart grids.
\newblock \emph{Proc. of the Nat. Acad. of Sci.}, 110(6), 2005--2010.

\bibitem[{Eriksson et~al.(2017)Eriksson, Modig, and Westberg}]{SVK17}
Eriksson, R., Modig, N., and Westberg, A. (2017).
\newblock Fcr-n design of requirements.
\newblock Technical report, ENTSO-e.
\newblock Accessed: 2025-11-06.

\bibitem[{Fletcher and Aasaraai(1989)}]{fletcher1989disturbance}
Fletcher, L.R. and Aasaraai, A. (1989).
\newblock On disturbance decoupling in descriptor systems.
\newblock \emph{SIAM J. on Contr. and Opti.}, 27(6), 1319--1332.

\bibitem[{Glover et~al.(2010)Glover, Sarma, and Overbye}]{Glover}
Glover, J.D., Sarma, M.S., and Overbye, T.J. (2010).
\newblock \emph{Power System Analysis and Design}.
\newblock Cengage Learning.

\bibitem[{Lebon and Altafini(2026)}]{lebon2026geometric_tac}
Lebon, L.C.G. and Altafini, C. (2026).
\newblock Geometric control theory over networks: Minimal node cardinality
  disturbance decoupling problems.
\newblock \emph{IEEE Transactions on Automatic Control}.

\bibitem[{Lindberg and Pates(2022)}]{lindberg2022fundamental}
Lindberg, J. and Pates, R. (2022).
\newblock Fundamental limitations on the control of lossless systems.
\newblock \emph{IEEE Contr. Sys. Lett.}, 7, 157--162.

\bibitem[{Machowski et~al.(2020)Machowski, Lubosny, Bialek, and Bumby}]{Mach20}
Machowski, J., Lubosny, Z., Bialek, J., and Bumby, J. (2020).
\newblock \emph{Power System Dynamics: Stability and Control}.
\newblock Wiley.

\bibitem[{Pai(1989)}]{pai89}
Pai, M. (1989).
\newblock \emph{Energy Function Analysis for Power System Stability}.
\newblock Kluwer Academic Publishers.

\bibitem[{Pasqualetti et~al.(2013)Pasqualetti, D{\"o}rfler, and
  Bullo}]{pasqualetti2013attack}
Pasqualetti, F., D{\"o}rfler, F., and Bullo, F. (2013).
\newblock Attack detection and identification in cyber-physical systems.
\newblock \emph{IEEE Trans. Autom. Contr.}, 58(11), 2715--2729.

\bibitem[{Pates(2022)}]{Pates22}
Pates, R. (2022).
\newblock Passive and reciprocal networks: From simple models to simple optimal
  controllers.
\newblock \emph{IEEE Contr. Sys. Mag.}, 42(3), 73--92.
\newblock \doi{10.1109/MCS.2022.3157116}.

\bibitem[{Simpson-Porco et~al.(2013)Simpson-Porco, D{\"o}rfler, and
  Bullo}]{simpson2013synchronization}
Simpson-Porco, J.W., D{\"o}rfler, F., and Bullo, F. (2013).
\newblock Synchronization and power sharing for droop-controlled inverters in
  islanded microgrids.
\newblock \emph{Automatica}, 49(9), 2603--2611.

\bibitem[{Skardal and Arenas(2015)}]{skardal2015control}
Skardal, P.S. and Arenas, A. (2015).
\newblock Control of coupled oscillator networks with application to microgrid
  technologies.
\newblock \emph{Sci. Adv.}, 1(7), e1500339.

\bibitem[{Wonham and Morse(1970)}]{wonham1970decoupling}
Wonham, W.M. and Morse, A.S. (1970).
\newblock Decoupling and pole assignment in linear multivariable systems: a
  geometric approach.
\newblock \emph{SIAM J. on Contr.}, 8(1), 1--18.

\bibitem[{Yin et~al.(2011)Yin, Mehta, Meyn, and
  Shanbhag}]{yin2011synchronization}
Yin, H., Mehta, P.G., Meyn, S.P., and Shanbhag, U.V. (2011).
\newblock Synchronization of coupled oscillators is a game.
\newblock \emph{IEEE Trans. Autom. Contr.}, 57(4), 920--935.

\end{thebibliography}
\bibliographystyle{ifacconf}

\end{document}